\documentclass[12pt]{iopart}
\pdfoutput=1
\usepackage{iopams}
\usepackage{amssymb}
\usepackage[pdftex]{graphicx}
\usepackage{url}
\usepackage{algorithm}
\usepackage[noend]{algpseudocode}
\usepackage{amsthm} 

\newcommand{\set}[1]{#1} 

\newtheorem{theorem}{Theorem}

\date{\today}

\begin{document}

\title{Temporal motifs in time-dependent networks}
\author{Lauri Kovanen$^1$, M\' arton Karsai$^1$, Kimmo Kaski$^1$,  J\'anos Kert\'esz$^{1,2}$ and Jari Saram\" aki$^1$}
\address{$^1$ Department of Biomedical Engineering and Computational Science, Aalto University School of Science, P.O. Box 12200, FI-00076, Finland}
\address{$^2$ Institute of Physics,  BME, Budapest, Budafoki \'ut 8., H-1111, Hungary}
\ead{lauri.kovanen@aalto.fi}

\begin{abstract}
Temporal networks are commonly used to represent systems where connections between elements are active only for restricted periods of time, such as networks of telecommunication, neural signal processing, biochemical reactions and human social interactions. We introduce the framework of \emph{temporal motifs} to study the mesoscale topological-temporal structure of temporal networks in which the events of nodes do not overlap in time. Temporal motifs are classes of similar event sequences, where the similarity refers not only to topology but also to the temporal order of the events. We provide a mapping from event sequences to colored directed graphs that enables an efficient algorithm for identifying temporal motifs. We discuss some aspects of temporal motifs, including causality and null models, and present basic statistics of temporal motifs in a large mobile call network.
\end{abstract}


\maketitle

\section{Introduction}

The network approach, where interacting elements are represented as nodes and interactions as edges, has become widely used in the study of complex systems \cite{NewmanBarabasiWatts2009,NewmanBook2010}. Although this approach unquestionably discards many details, it has turned out to provide much insight into the function and dynamics of the systems in question. Many large networks display similar properties on the global scale, such as broad degree distributions, short path lengths, and abundance of triangles; on the mesoscopic level, complex networks often display community structure \cite{Fortunato2010}. There is much variation in the mesoscale structure of different networks \cite{Lancichinetti2010,OnnelaTaxonomy2010}, reflecting different underlying functional and dynamical mechanisms. Such differences can also be observed in the relative significance of \emph{motifs}, sets of small topologically equivalent subgraphs \cite{Shen-Orr_NatureGenetics2002,Milo_Science2002}. The concept of motifs has also been generalized to unweighted \cite{Shen-Orr_NatureGenetics2002,Milo_Science2002} and weighted \cite{Onnela_PRE_2005} networks.

Static networks are often time-aggregates of systems where connections are not continuously active but established only during limited periods of time. This temporal aspect turns out to be crucial for processes like spreading of information and electronic viruses in communication networks \cite{Holme_PRE2005, Kossinets2008, Centola_Science2010, VazquezBurstySpreadingPRL2007, Iribarren2009,   Karsai2011, Miritello2011}, epidemiological applications \cite{Rocha2011,Lee2010}, and signal processing in the brain (see, \emph{e.g.}, \cite{valencia2008, Dimitriadis2010}). Sometimes temporal aspects such as link activation frequencies are incorporated in the static network representation as link weights, which are then assumed to affect dynamic processes like spreading in probabilistic mean field manner. However, it has recently become clear that temporal inhomogeneities not captured by this approach have an important effect on many processes \cite{VazquezBurstySpreadingPRL2007, Iribarren2009,   Karsai2011, Miritello2011, Rocha2011}.

In this article we use the \emph{temporal networks} approach (see, \emph{e.g.}, \cite{cui2010} and \cite{Holme_arXiv2011} for a review) to study the details of link activations without projecting out the temporal dimension. While static networks consist of nodes and edges, temporal networks consist of nodes and \emph{events}. A (directed) event $e_i = (n_{i,1}, n_{i,2}, t_i, \delta_i)$ connects the two nodes $n_{i,1} \to n_{i,2}$ only during the time interval from $t_i$ to $t_i+\delta_i$. Our presentation uses directed events for generality, but undirected events can be used with minimal changes. We restrict ourselves to the case where nodes cannot participate in simultaneous events, \emph{i.e.}, at any given time at most one event can be assigned to a node.

It is reasonable to expect temporal networks to have mesoscale structures both in topology and time. These structures are likely to reflect the function of the system even better than mesoscale structures in static networks and thus their characterization can improve our understanding of various complex systems, from the nature of human social interactions and information processing by groups to temporal patterns that determine the outcomes of dynamical processes like spreading. Nicosia \emph{et al.}~have studied strongly and weakly connected components in temporal networks \cite{Nicosia2011}. Braha and Bar-Yam \cite{braha_bar_yam} have studied motifs in static snapshot networks, aggregated over one day of email data, and found that dense subgraphs are overrepresented. Bajardi \emph{et al.}~\cite{Bajardi_arXiv2011} have defined dynamical motifs as sequences of connected events belonging to adjacent time windows of network aggregation. In essence, these are time-respecting paths \cite{kempe_etal,Holme_PRE2005}, that is, linear chains of events, and are thus different from the temporal motifs discussed in this paper. Zhao \emph{et al.}~\cite{Zhao_CIKM2010} have studied communication motifs in electronic social networks with an approach that has some similarity to the approach we take in this paper -- they, too, consider subsets of communication events where the time between consecutive events sharing a node is within a chosen threshold time. However, in their analysis, the time dimension of such temporal subgraphs appears to have been projected out by projecting the patterns into static subgraphs.

The \emph{temporal motifs} we introduce here can be used to study the full mesoscale topological-temporal structure of temporal networks. We also present an efficient algorithm for identifying all temporal motifs in empirical data sets on temporal graphs, that is, time-stamped sequences of events between nodes. In static networks the motifs are---in a very general sense---defined as classes of isomorphic, connected subgraphs. We define temporal motifs analogously, first by defining connected subgraphs in temporal networks and then by extending the definition of isomorphism such that it also takes into account the temporal information in these subgraphs.

As an example, in a social communication network one might detect an event sequence where Alice calls Bob, who then calls Carol and Dave. A similar sequence might be observed to often take place between the same people, as well as between other sets of four individuals. All these sequences are members of the same class, which we call a temporal motif. In genetic regulation data the event sequence would correspond to regulatory interactions switching on and off as the intercellular system performs its function. In addition to providing insight into the operation of the system under study, temporal motifs allow studying similarities and differences of temporal networks, as originally proposed for static motifs in \cite{Milo_Science2002}. In addition they may help in building models of network evolution \cite{Kumpula_PRL2007}.

We start with a formal definition of temporal motifs in Section \ref{sec:definitions}, and then cover the main ideas of the identification algorithm. In Section \ref{sec:generalizations} we discuss some useful generalizations and show how they can also be implemented efficiently. Section \ref{sec:evaluation} glances at different methods for evaluating the significance of the observed motif counts. Finally, in Section \ref{sec:results} we use the methods to identify temporal motifs in a large temporal network constructed from mobile phone data and discuss the insights the motifs provide. A detailed account on all algorithms is provided in the Appendix.


\section{Temporal motifs}
\label{sec:definitions}

Static motifs are classes of isomorphic subgraphs. While there is variation in exactly what kind of subgraphs are studied, it is practically always required that these subgraphs be connected. For static graphs connectivity means that there is a path between all pairs of nodes, or equivalently that there is a sequence of mutually \emph{adjacent} edges between all pairs of edges, where two edges are adjacent if they have one node in common.

In temporal networks the definition of adjacency should intuitively also include time; two calls made by the same person a month apart are hardly close to each other. We consider two events $\Delta t$\emph{-adjacent} if they have at least one node in common and the time difference between the end of the first event and the beginning of the second event is no longer than $\Delta t$. Equivalently, two events are $\Delta t$\emph{-connected} if there exists a sequence of events $e_i=e_{k_0} e_{k_1} \ldots e_{k_n}=e_j$ such that all pairs of consecutive events are $\Delta t$-adjacent.\footnote{Note that this sequence does not need to be a journey, i.e.\ the events need not be temporally ordered.}

Using these definitions, a \emph{connected temporal subgraph} consists of a set of events such that all pairs of events in it are $\Delta t$\emph-connected. This ensures that subgraphs are connected both topologically and temporally\footnote{This is different from the approach taken by Zhao \emph{et al.} \cite{Zhao_CIKM2010}, where temporal subgraphs may be topologically disconnected.}. While this definition could already be used as a basis of temporal motifs, it suffers from the same shortcoming as its static cousin: in some simple cases the number of connected subgraphs explodes. For example an $n$-star where all events take place within $\Delta t$ contains ${n \choose k}$ connected temporal subgraphs with $k$ events, which would make the resulting motif statistics difficult to interpret in any intuitive fashion.

With static motifs the most common restriction is to require the subgraphs to be both connected and induced, i.e. require that they include all edges between the nodes in the subgraph. While this choice does reduce the number of subgraphs and makes it easier to interpret the resulting motifs, it unfortunately fails to solve the problem with the $n$-star above.

With temporal networks our choices are not as restricted. One alternative is to consider only those connected subgraphs where all $\Delta t$-connected events of each node are consecutive. This not only solves the problem with the $n$-star---we now get $n-k+1$ subgraphs with $k$ events---but also offers an intuitive interpretation: each subgraph takes into account all relevant events for each node within the time span covered for that node, in the sense that no events can be skipped. We call connected subgraphs that satisfy this requirement \emph{valid temporal subgraphs} and denote them by $\set{E}^*$. Figure \ref{fig:temporal_subgraphs_illustration} illustrates the concept.

\begin{figure}[t] \centering
\includegraphics[width=12cm]{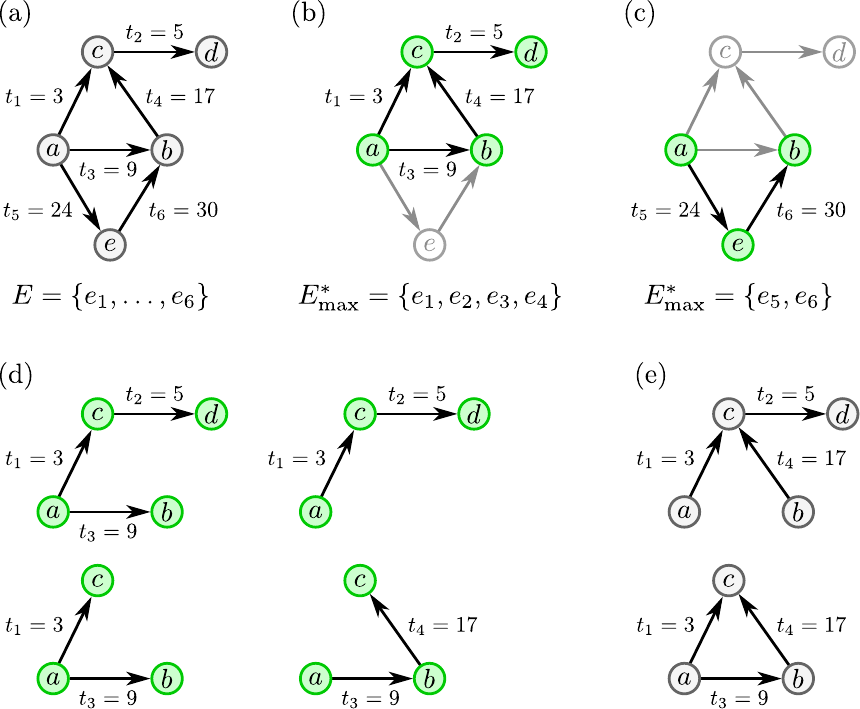}
\caption{\label{fig:temporal_subgraphs_illustration} \textbf{(a)} An example event data set $\set{E}$ with six events. Durations have been omitted for simplicity. With $\Delta t = 10$ there are two maximal subgraphs, shown in \textbf{(b)} and \textbf{(c)}. \textbf{(d)} Valid subgraphs contained in the maximal subgraph in (b). In addition to these the maximal subgraph itself and all unit subgraphs are valid subgraphs. The maximal subgraph in (c) does not contain other valid subgraphs than the maximal and unit subgraphs. \textbf{(e)} Event sets that are contained in (b) but are not valid subgraphs: the upper one because it is not $\Delta t$-connected, the lower one because it does not include all consecutive $\Delta t$-connected events of node $c$.}
\end{figure}

\emph{Temporal motifs} are now defined as classes of isomorphic valid subgraphs, where the isomorphism is taken to include also the similarity of the temporal order of events. Accordingly, two temporal subgraphs are isomorphic if they are topologically equivalent and the order of their events is identical. In cases where the requirement for the identity of the full order of events is too strict, it can easily be weakened. This is discussed in Section \ref{sec:generalizations}.

Some special temporal motifs are worth mentioning. The unit set $\set{E}^* = \{ e_i \}$ is trivially a valid subgraph for all events, and hence the smallest temporal motif contains only one event. For every event $e_i$ there is a unique maximal subgraph $\set{E}_{\max}^*$ that contains $e_i$ and in which all event pairs are still $\Delta t$-connected. The maximal subgraph is always also a valid subgraph. When motifs are based only on maximal subgraphs they are called \emph{maximal motifs}.

The presented definition for temporal subgraph is meaningful only when each node is involved in at most one event at a time. When overlapping events are allowed, the large number of different situations that can arise in the most general case makes it more difficult to define temporal subgraphs in such a way that the results could still be easily interpreted. Yet the prospects of such a definition are enticing, as it would allow for the exploration of almost any temporal network data, for example transportation networks \cite{Pan11} and time-varying brain functional networks \cite{ValenciaBrain}.

\section{Algorithm for identification of temporal motifs}
\label{sec:identification}

Because maximal subgraphs are temporally separated from all other events by at least time $\Delta t$, all subgraphs are fully contained in some maximal subgraph. Based on this observation the process of identifying all temporal motifs in a given event set $\set{E}$ can be separated to three parts:
\begin{enumerate}
\item Find all maximal connected subgraphs $\set{E}_{\max}^*$.
\item Find all valid subgraphs $\set{E}^* \subseteq \set{E}_{\max}^*$.
\item Identify the motif corresponding to $\set{E}^*$.
\end{enumerate}

To find the maximal subgraph where $e_i$ belongs to, we start from $e_i$ and iterate forward and backward in time to find all $\Delta t$-adjacent events; this process is then repeated recursively with all new events encountered. Assuming the $\Delta t$-adjacent events can be found in constant time, the time complexity of this step is $O(|\set{E}_{\max}^*|)$. Since the same maximal set is discovered starting from any event in it, the total time complexity of this part is $O(|\set{E}|)$.

For the second part, consider an undirected graph $G$ where the vertices corresponds to events in $\set{E}_{\max}^*$ and there is an edge between two vertices if those events are $\Delta t$-adjacent and consecutive for either node. Now each valid subgraph contained in $\set{E}_{\max}^*$ corresponds to some connected vertex set of $G$ (see \ref{appendix:subgraphs} for proof), and the problem of finding all valid temporal subgraphs reduces to identifying all induced subgraphs of $G$ and checking that the events of each node are consecutive. The algorithm for identifying valid subgraphs is given in \ref{appendix:subgraphs}.

Identifying the motif for subgraph $\set{E}^*$ requires solving the isomorphism problem such that we also include information about the order of the events. We do this by mapping all relevant information into a directed and coloured\footnote{In a coloured graph each vertex has an additional property called colour. We represent both actual nodes and events as vertices and need colours to distinguish the two.} graph as illustrated in Figure \ref{fig:motif_mapping}, for which the isomorphism can be readily solved with existing algorithms. In practice we calculate for this graph its \emph{canonical form}, a labeling of vertices that is identical for all isomorphic graphs, so that we can easily tell if two valid subgraphs correspond to the same motif. Finding the canonical form is a non-trivial task, but many efficient algorithms have been developed for this purpose; the one we used is called \emph{bliss} and described in \cite{Junttila_ALENEX2007}.

As a final step, to make temporal motifs more accessible we convert the information about the order of events back into plain integers. Figure \ref{fig:motif_mapping}e shows a concise presentation of the motif corresponding to the original temporal subgraph in Figure \ref{fig:motif_mapping}a.

\begin{figure}[t] \centering
  \includegraphics[width=14cm]{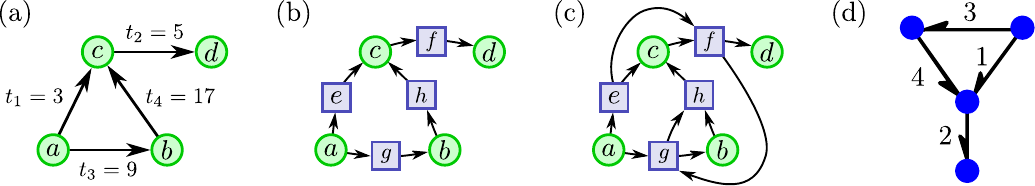}
  \caption{\label{fig:motif_mapping} Illustration of the algorithm for     identifying temporal motifs. \textbf{(a)} A valid subgraph $\set{E}^*$ with four events. \textbf{(b)} A vertex is created for each event and edges are added to connect them to the corresponding nodes. Colours are used to distinguish between the two types of vertices; the labels of the event vertices are arbitrary. \textbf{(c)} Directed edges are created between event vertices to denote their order: from the first event ($t_1=3$) to the second ($t_2=5$), from the second to the third ($t_3=9$) and from the third to the fourth ($t_4=17$). When durations are included we use the order of the starting times. A canonical labeling is then calculated for this graph; all temporal subgraphs with that are isomorphic at this stage will yield the same canonical labeling. \textbf{(d)} A concise presentation for the temporal motif. The numbers next to edges denote the order of the events. Note that the numbers are always on the side of the arrow heads.}
\end{figure}

\section{Flow motifs and partial order of events}
\label{sec:generalizations}

\begin{figure}[t] \centering
  \includegraphics[width=0.95\textwidth]{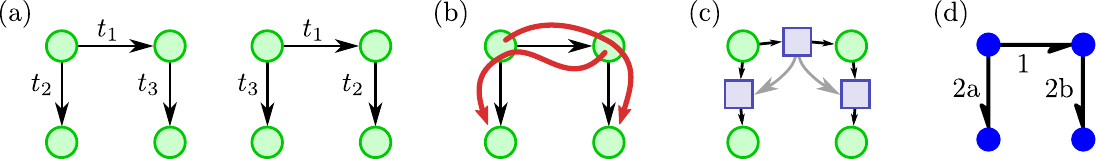}
  \caption{\label{fig:journeys_motivation} \textbf{(a)} Two valid subgraphs that differ only in the mutual order of events $e_2$ and $e_3$. \textbf{(b)} If the events are mobile phone calls, the possible flow of information (red arrows) is identical in the two subgraphs. The mutual order of $e_2$ and $e_3$ is irrelevant. \textbf{(c)} The temporal flow motif corresponding to both event sets in (a) where the only requirement is that $e_1$ takes place before the other two events. \textbf{(d)} Compact notation for the temporal motif in (c). As described in the text, with this notation $1 = (1,\emptyset) < (2,\{a\}) = 2a$ and $1 < 2b$, but the order of $2a$ and $2b$ is undefined.}
\end{figure}

Assuming that the node colours are denoted by integers---as is often the case---we could have used the colours to mark the order of events in Figure \ref{fig:motif_mapping}c instead of putting in additional links. The edge notation, however, has another benefit: it can be used to denote a \emph{partial order} of events. Unlike \emph{total order}, partial order does not necessarily define the order of all pairs. For example, an order where $e_i$ takes place before both $e_j$ and $e_k$, but the mutual order of $e_j$ and $e_k$ remains undefined, is a partial order and as such cannot be represented with integer labels.

To see why this is useful, consider mobile phone communication where information flows both ways---both can talk regardless of who placed the call---and the two temporal subgraphs in Figure \ref{fig:journeys_motivation}a that differ only in the order of the last two events. If we are only interested in the flow of information, the two subgraphs are identical because they allow the same flows, shown in Figure \ref{fig:journeys_motivation}b. In general such flows are known as \emph{time-respecting paths} or \emph{journeys} \cite{Holme_PRE2005,kempe_etal, BBXuan03}: a sequence of events $e_1e_2 \ldots e_n$ such that consecutive events are adjacent and $t_i < t_{i+1}$ $\forall i=1,\ldots,n-1$.

In a \emph{flow motif} the order of two events is restricted if and only if it is relevant to the flow pattern in the subgraphs, that is, only when reversing the mutual order of two events would either create a new journey or remove an existing one. Because journeys must progress via adjacent events it is enough to place restrictions on the order of adjacent events; all longer journeys will be automatically included. If the flow is undirected, such as information flow during phone calls, preserving journeys (and not making new ones) corresponds to restricting the order of all adjacent events as shown in Figure \ref{fig:journeys_motivation}c. In the case of directed flow we would only restrict the order of events that meet head-to-tail; no flow is possible if the events meet either head-to-head or tail-to-tail.

If the events have a partial order we can of course no longer use integers to denote this order as was done in Figure \ref{fig:motif_mapping}d. Arbitrary sets could be used to represent partial orders by defining $x < y~\Leftrightarrow~x \subset y$, but they would render the most common simple motifs unnecessarily complicated. We propose a notation that uses sets when necessary but falls back to plain integers when possible. We label events with pair $(r,\set{s})$ where $r \in \mathbb{N}$ and $s$ is a set, and define order relation as
\[
(r_i,\set{s}_i) < (r_j,\set{s}_j) ~ \Leftrightarrow ~ r_i < r_j \land \set{s}_i \subseteq \set{s}_j~~~.
\]
By choosing $\set{s}_i = \emptyset~\forall\,i$ whenever possible the notation reduces to a comparison of integers because in this case $\set{s}_i \subseteq \set{s}_j~\forall i,j$. To make the notation more compact we write $(r,s)$ as `$rs$' or simply `$r$' if $s = \emptyset$. For example in Figure \ref{fig:journeys_motivation}d the label `$1$' corresponds to $(1,\emptyset)$ and `$2a$' to $(2,\{a\})$. In \ref{appendix:labels} we present an algorithm for finding such labels for any partial order of events.

\section{Evaluation of motif statistics}
\label{sec:evaluation}

The standard interpretation of a static motif count, i.e. the number of subgraphs in the motif, is presented in terms of a null model  \cite{Shen-Orr_NatureGenetics2002, Milo_Science2002}. The null model is usually a conditionally randomized version of the empirical network, e.g. a configuration model with the same degree sequence as the empirical network. If for some motif the count significantly exceeds that of the null model, the hypothesis (i.e., lack of correlations reflected in the motif) is rejected and the motifs are considered to be structurally significant. However, as pointed out in \cite{Artzy-Randrup_Science2004}, the proper choice of the null model is non-trivial. If the null model is far from having any realistic features, then it is no wonder that it is rejected but this plain fact does not tell anything about the nature of the correlations. The standard z-score analysis compares the difference of the empirical motif count and the average value from the null model to the variance of the latter. This Gaussian assumption about the null model has no a priori justification.

This problem is even more severe for temporal motifs. Here the most obvious randomized reference is time-shuffling \cite{Bajardi_arXiv2011}: given a random permutation $\phi$ of events we count the occurrence of motifs in time-shuffled data set where event $e_i$ occurs at time $t_{\phi(i)}$. Unlike in the time-shuffled reference, in most complex systems temporal distributions are far from Poissonian and contain strong temporal correlations \cite{Karsai2011, Miritello2011, Rocha2011, Bajardi_arXiv2011}. The situation is improved if we use parametrized null models, where in some limit the empirical situation is restored. Then we can hope that by monitoring the parameter dependence of the deviations from the null model we can learn about the nature of the correlations.

Another intuitive choice is to compare the occurrence of motifs to a time-reversed reference \cite{Bajardi_arXiv2011}. Since causality depends on the direction of time but correlation does not, this comparison should highlight motifs whose occurrence at least partially results from causality. On the other hand, if a motif is abundant only because of correlations, it should be equally common in both the data and the time-reversed reference. Note that it is not necessary to explicitly construct the time-reversed reference: the occurrence of a motif in the time-reversed data is equal to the occurrence of a time-reversed motif in the actual data.

Considering the problems with null models, it seems to be important to compare parts of the data with \emph{itself}. If there are different types of nodes and events, we can study whether the occurrence of temporal motifs differs between them. Also, we can always study the occurrence of motifs at different times. In this way we would gain information about the relative weights of the motifs without any reference to arbitrary null models.

When analyzing motif counts we need to take into account that they are trivially correlated with average activity and correlations of adjacent events. To get more insight into the occurrence patterns of temporal motifs we suggest looking at the relative occurrence of different motifs. Suppose that we have two sequences of motif counts---for example the counts of all 3-event temporal motifs in the empirical data and the reference---and the relative frequencies of the $i^{\textrm{\small th}}$ motif are $p_i$ and $q_i$. The symmetrized Kullback-Leibler divergence measures the relative entropy of these two distributions and is defined as
\[
D_{\textrm{\small KL}}(\{ p\},\{ q\}) = \sum_{i=1}^n p_i \log \frac{p_i}{q_i} +  \sum_{i=1}^n q_i \log \frac{q_i}{p_i}~~~
\]
provided that $p_i > 0$ and $q_i > 0$ $\forall i$ (we exclude motifs that are not present in either).

The Kullback-Leibler divergence places more weight on common motifs, and even large relative differences in the rarest motifs do not change the value too much. Kendall`s $\tau$, on the other hand, measures the similarity of the ordered sequences and places an equal weight on all the motifs regardless of their count. Given two motif sequences of length $n$, both sorted by count, Kendall`s $\tau$ is defined as
\[
\tau = \frac{R^+ - R^-}{\frac{1}{2}n(n-1)}
\]
where $R^+$ is the number of motif pairs that are in the same order in both sequences and $R^-$ the number of motif pairs in different order. The value $\tau = 1$ is reached when the two sequences are identical, and $\tau = -1$ when they are in opposite order.

\section{Results}
\label{sec:results}

We use temporal motifs to study the short time-scale structure of mobile phone calls of a single European mobile phone operator. The data covers a period of 120 days, but we exclude motifs that occur entirely on the first or the last day of this period to remove possible bias caused by the limits of the data. The remaining data contains $320$ million mobile phone calls between nearly $9$ million customers. The data has been mutualized by removing all events on unidirectional edges, i.e. we require that the communication is reciprocated on each edge. The time window is $\Delta t = 10\textrm{ min}$ except in Figure \ref{fig:maximal_motif_distributions} where other time windows are explored. With $\Delta t = 10\textrm{ min}$, 35 \% of events are $\Delta t$-adjacent to at least one other event and hence non-trivial temporal motifs are not all that rare. All results with time-shuffled references have been averaged over 5 independent runs.

\begin{figure}[t] \centering
 \includegraphics[width=16cm]{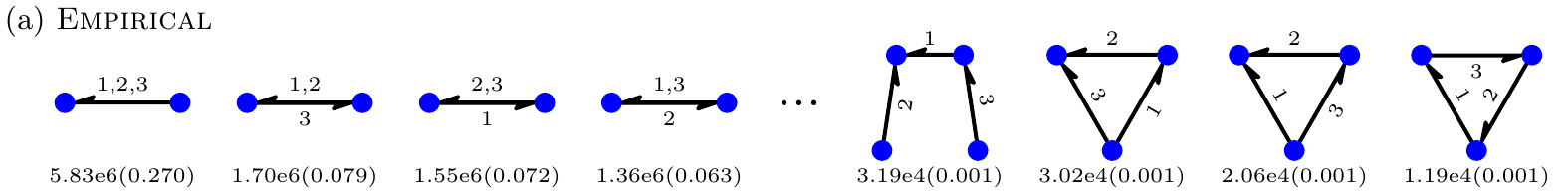}\\
\vspace{1em}
 \includegraphics[width=16cm]{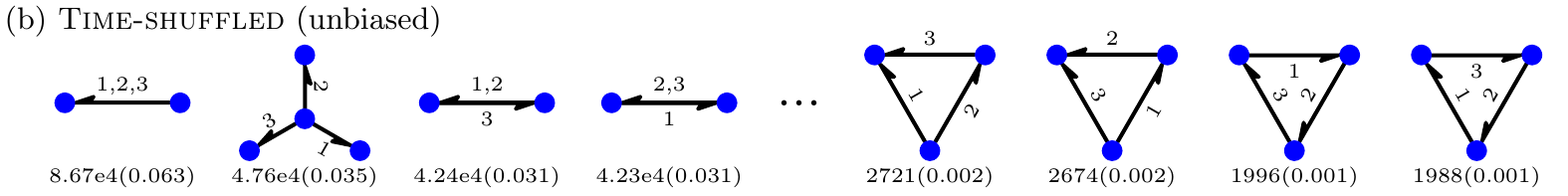}\\
\vspace{1em}
 \includegraphics[width=16cm]{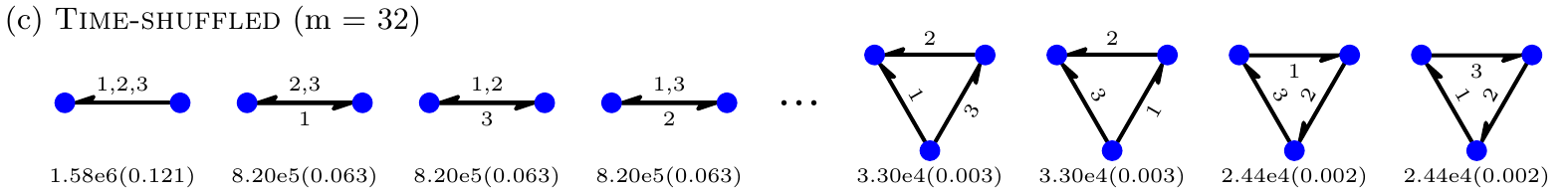}
 \caption{\label{fig:most_and_least_common_motifs} The four most    common (on left) and least common (on right) motifs in \textbf{(a)}    the empirical data, \textbf{(b)} unbiased time-shuffled reference    and \textbf{(c)} the biased reference with bias strength $m=32$.    The values below each motif denote the total count and, in parenthesis, the fraction out of all motifs with three events.}
\end{figure}

Figure \ref{fig:most_and_least_common_motifs}a shows the four most and least common 3-event temporal motifs (there are 68 3-event motifs in total) in the
data, and Figure \ref{fig:most_and_least_common_motifs}b the same in the time-shuffled reference. Unsurprisingly, the number of non-trivial motifs in the reference is lower---only 8.6 \% of events are $\Delta t$-adjacent to some other event---but the two cases still appear qualitatively similar. The most common motifs illustrate the bursty nature of the mobile phone data, while the least common motifs are triangles even though triangles are often considered to be the building blocks of social networks. The distribution of different motifs is more balanced in the reference: in the empirical data the most common 3-event motif makes up 27 \% of all 3-event motifs, but only 6.3 \% in the time-shuffled reference.

To make the comparison more interesting, we add a bias to the time-shuffling that favors shorter inter-event times and therefore increases the number of motifs. The shuffling is done using Markov chain Monte Carlo sampling, which is also necessary to enforce the condition that each user is involved in at most one event at a time. In the unbiased case each step consists of selecting two events uniformly at random, $e_i$ and $e_j$, and switching their times if this does not result in overlapping events for any of the (at most) four nodes involved. To create a single randomized reference we make $5|E|$ such switches, which equals on average 10 switches per event.

To introduce a bias, instead of picking only two events at each step we randomly select one target event $e_i$ and $m \geq 1$ candidates,  $(e_{j_1},\ldots,e_{j_m})$, and then make a switch with the candidate that places $e_i$ closest to its new adjacent events. To measure this closeness we use the geometric average of time differences to the temporally closest adjacent events.\footnote{As we are only interested in the order of these averages and not their exact values, comparing geometric averages is equal to comparing the arithmetic average of logarithms of time differences. This puts more importance to small time differences than plain arithmetic average.} The parameter $m$ controls the bias strength: the more candidates there are, the more likely it is to find one close to $e_i$. Setting $m=1$ gives the normal unbiased randomization.

Figure \ref{fig:most_and_least_common_motifs}c shows the most and least common motifs in the biased reference with $m=32$. This reference naturally has higher motif counts than the unbiased reference, but the total number of 3-event motifs is still only 60 \% of that of the empirical data. Perhaps surprisingly, the least common motif is now twice as common as in the empirical data. In the empirical data this motif is uncommon partly because the events take place in a non-causal order, whereas the order has little significance in the reference. Furthermore, because this kind of subgraph takes place primarily due to correlations, it is likely that the nodes have other events at approximately the same time. If these events take place between those in the triangle, the subgraph would no longer be valid (see lower subgraph in Figure \ref{fig:temporal_subgraphs_illustration}e). In the references the maximal subgraphs are smaller, which makes it less likely that such interfering events would destroy the validity. The bias makes triangles more common while keeping the maximal subgraphs small, and therefore the triangles are more often valid.

\begin{figure}[t] \centering
\mbox{
\includegraphics[width=11.7cm]{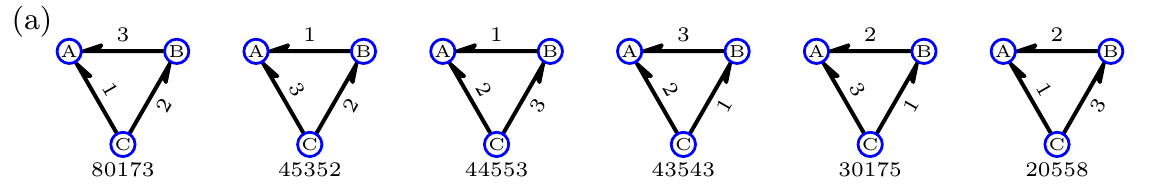}
\includegraphics[width=3.9cm]{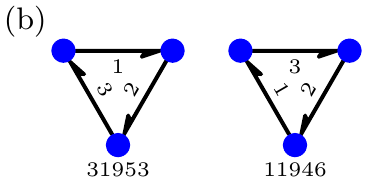}}
\caption{\label{fig:triangle_motifs} The two different kinds of directed triangle motifs with 3 events. Both groups have been ordered by count in the empirical data that is also shown below the motifs. All motifs in (a), as well as those two in (b), differ only in the order of events.}
\end{figure}

As a further example clarifying this point, we present in Figure \ref{fig:triangle_motifs} all motifs based on the different directed triangles with 3 events. The six motifs in Figure \ref{fig:triangle_motifs}a would by equally common in the time-shuffled reference, but in the empirical data we observe a 4-fold difference between the most and least common triangle. There are two factors that explain this: \emph{burstiness} and \emph{causality}. Burstiness appears in the fact that in the four most common motifs the two calls made by $C$ are consecutive; in the two least common motifs they are not. Causality is most apparent when comparing the most and the least common motif. In the most common motif the caller of the second call (C) knows about the first call (because he made it himself), and the caller of the third call (B) could know about both previous calls. In the least common motif the caller of the second call (B) cannot know about the first one, and the caller of the third call (C) cannot know about the call made by B. The most common motif is both bursty and causal, while the least common is neither.

Causality is an obvious explanation also for the counts in Figure \ref{fig:triangle_motifs}b: the triangle where events could cause one another is three times as common as the other one. Note that these two motifs are time-reversals of each other, i.e. if the time were reversed, each motif of the first kind would turn into the second, and vice versa.

\begin{figure}[t] \centering
  \includegraphics[width=12cm]{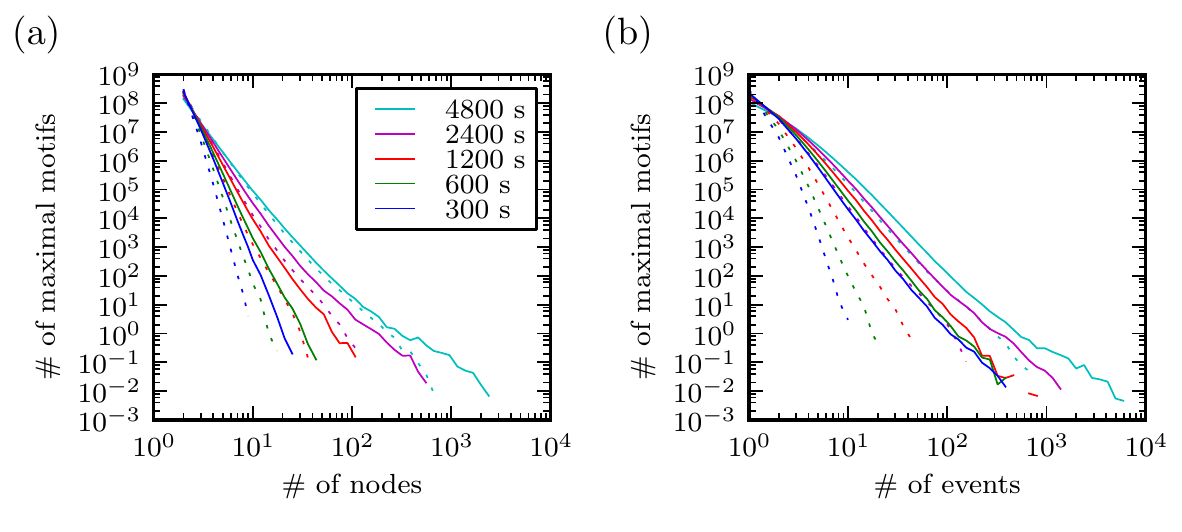}
  \includegraphics[width=12cm]{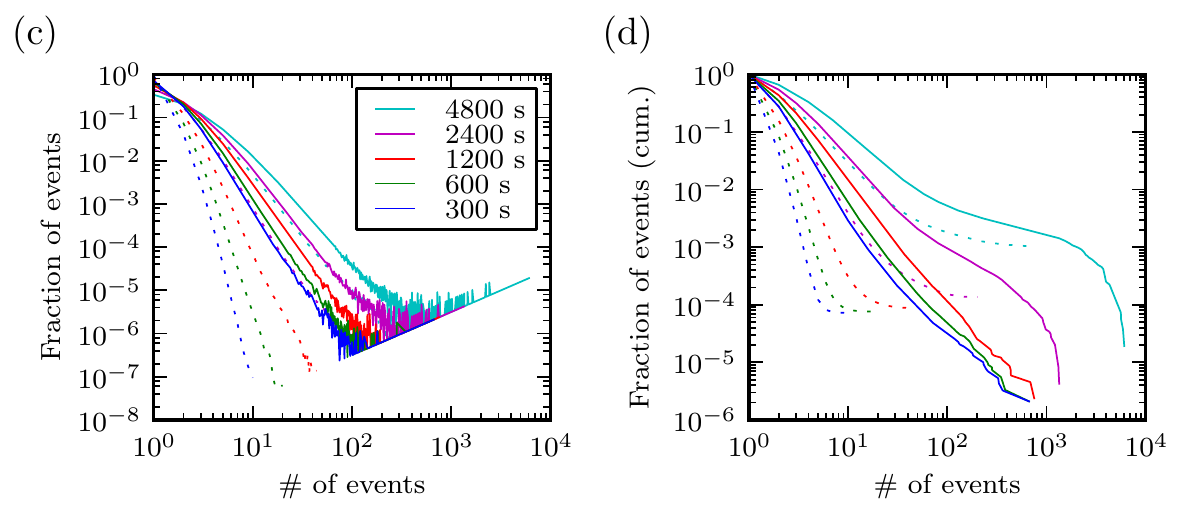}
  \caption{\label{fig:maximal_motif_distributions} Number of maximal motifs of different size when the size is measured by \textbf{(a)} the number of nodes and \textbf{(b)} the number of events in the motif. In both plots the values larger than 10 have been binned with logarithmic bins using factor $1.2$. \textbf{(c)} Fraction of events in motifs of a given size, and \textbf{(d)} the corresponding cumulative distribution. In all plots the solid lines correspond to empirical data, dotted lines to the unbiased time-shuffled reference.}
\end{figure}

Figures \ref{fig:maximal_motif_distributions}a--b show the number of maximal motifs of different size for different values of $\Delta t$, measured either by the number of nodes or by the number of events in the motifs. The distributions are broad for all time windows, and those with larger $\Delta t$ are naturally broader. Figures \ref{fig:maximal_motif_distributions}c--d show the fraction of events in maximal motifs of different size. Comparing the distributions with $\Delta t = 1200$ and $2400$ suggests that between these values a giant temporal component is beginning to form. The distribution with $\Delta t = 1200$ is very close to a power-law, as both the density and cumulative distributions are straight lines. When $\Delta t = 2400$ the number of events contained in very large maximal motifs is starting to grow. Increasing the time window further beyond $\Delta t = 4800$ would at some point give birth to a giant temporal component: a large fraction of events would become $\Delta t$-connected.

\begin{figure}[t] \centering
  \includegraphics[width=14cm]{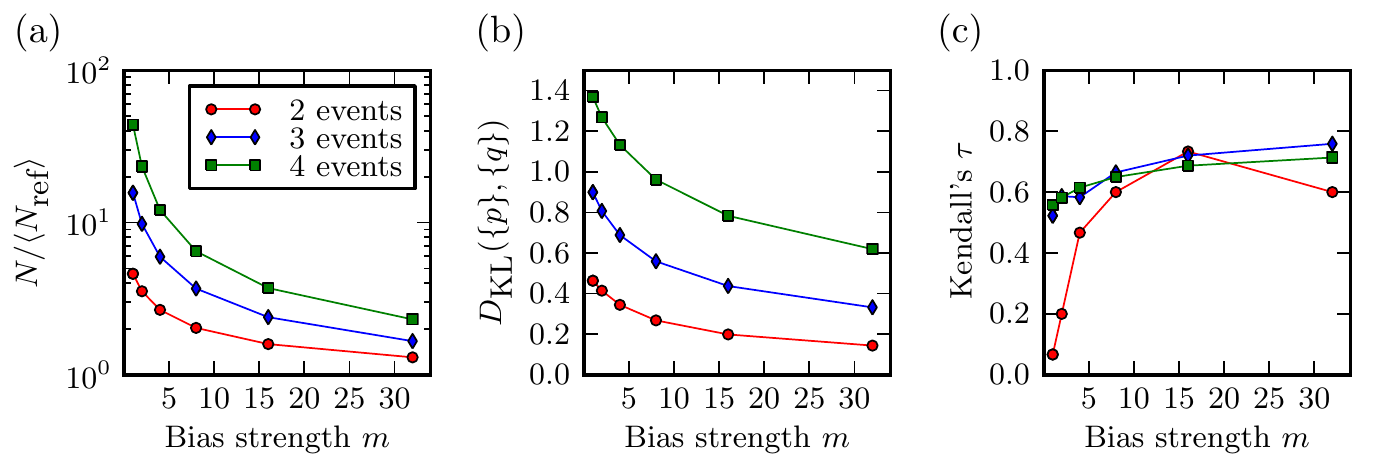}
  \caption{\label{fig:bias_results} \textbf{(a)} The ratio of total number of motifs with a given number of events in the actual data versus the time-shuffled reference. The lines correspond to motifs with 2, 3 and 4 events, from bottom to top. \textbf{(b)} The symmetric Kullback-Leibler divergence between motifs in the actual data and the time-shuffled reference. \textbf{(c)} Kendall's $\tau$ between motifs in the actual data and the time-shuffled reference. The motif counts of the references are averages of 5 different runs for each value of the bias strength.}
\end{figure}

Finally, Figure \ref{fig:bias_results}a shows that if we only look at the number of motifs, the biased references seem to approach the actual data as we increase the bias strength. Similar behaviour is seen in Figure \ref{fig:bias_results}b for the symmetrized KL divergence calculated between the actual data and the reference, and also for Kendall`s $\tau$ in Figure \ref{fig:bias_results}c. However, in Figure \ref{fig:most_and_least_common_motifs} we saw that with $m=32$ there are already motifs that are more common in the reference than in the real data; it is therefore not possible that the reference becomes identical to the data when the total number of motifs become equal. In addition, motifs with more events are relatively more common in the empirical data regardless of the bias. A qualitative difference between motif sequences remains even if we were able to match the total number of motifs.

\section{Conclusions}

In this paper we have introduced the concept of temporal motifs, and provide a mapping between temporal subgraphs and colored directed graphs that allows an efficient algorithm for their identification. Using this algorithm we can locate and make statistics about the main mesoscopic building blocks of temporal networks, which will carry great importance for understanding their functions and underlying mechanisms.

While the focus of this article is more on technical definitions, algorithms, and general aspects of the evaluation of motif counts, we have also presented some results on a huge temporal network based on mobile phone call data. Some conclusions can already be drawn. Of all temporal motifs with three events, the most common ones involve only two nodes. This is in accord with the independent finding that burstiness in human communication is mostly a link property \cite{Karsai_tobepub}. Another interesting---though not surprising---result is that the motifs which allow causal interpretations are more common. The fat-tailed distributions of maximal motifs are in agreement with observations about the correlations in the network \cite{Karsai2011}, although our present approach gives a more detailed insight about the mechanisms. Our initial results also show that the temporal motifs are common enough to have an impact on dynamics and too complex to be explained by simple temporal correlations. The occurrence of motifs is intuitively sensible, as they highlight two universal properties of human communication, namely burstiness and causality.

There are a number of directions to pursuit in the future. For example here we have ignored the case where nodes can have simultaneously multiple events. This generalization presents a challenge both in defining the valid subgraphs and in developing the algorithms for their identification. Further research is also needed to develop measures for analyzing the occurrence patterns of motifs.

The presented examples are far from being able to illustrate the full richness of phenomena that can be explored with temporal motifs. Currently we are in the course of investigating several temporal networks and hope that it this approach will be useful in a broad range of studies, even more so as the presented algorithm is able to handle large networks. As empirical temporal networks are becoming more and more common, we expect temporal motifs and their analysis to prove useful in many different fields of science.

\ack

The project ICTeCollective acknowledges the financial support of the Future and Emerging Technologies (FET) programme within the Seventh Framework Programme for Research of the European Commission, under FET-Open grant number: 238597. We acknowledge support by the Academy of Finland, the Finnish Centre of Excellence program 2006--2011, project no.\ 213470. JK is supported by the Finland Distinguished Professor (FiDiPro) program of TEKES.

We would like to thank Renaud Lambiotte for his comments and Albert-L\'{a}szl\'{o} Barab\'{a}si of Northeastern University for providing access to the unique mobile phone data set.



\appendix

\section{Finding temporal subgraphs}
\label{appendix:subgraphs}

The following result is used in Section \ref{sec:identification} to find temporal subgraphs inside maximal temporal subgraphs:
\begin{theorem}
  Let $G(\set{E}^*_{\max})$ be an undirected graph that has a   vertex for each event in $\set{E}^*_{\max}$ and every vertex is   connected to the next and previous $\Delta t$-adjacent event of both   nodes in that event (there are at most four such events). Then every   valid subgraph contained in $\set{E}^*_{\max}$ corresponds to a   connected subgraph of $G(\set{E}^*_{\max})$.
\end{theorem}
\begin{proof}
  Consider a valid subgraph $\set{E}^* \subseteq \set{E}^*_{\max}$ and   the corresponding vertex set in $G$. Because all event pairs in   $\set{E}^*$ are $\Delta t$-connected and the events of every node   are consecutive, there is at least one path between all vertex pairs   in this set. Therefore there is at least one connected subgraph of   $G$ that corresponds to $\set{E}^*$.
\end{proof}

\begin{figure}[t] \centering
\includegraphics[width=14cm]{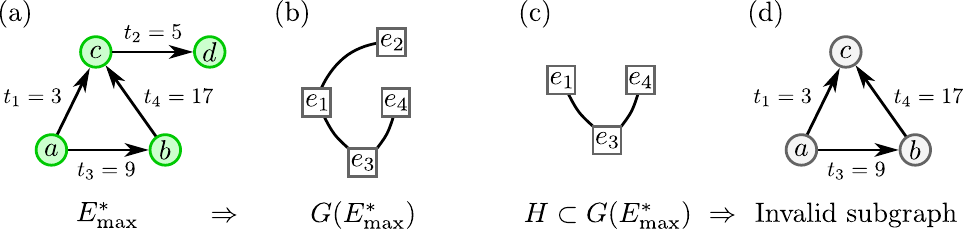}
\caption{\label{fig:connected_subgraphs} \textbf{(a)} An example of a maximal subgraph $\set{E}^*_{\max}$ with $\Delta t = 10$ and \textbf{(b)} the corresponding undirected graph $G(\set{E}^*_{\max})$ used to identify all valid subgraphs contained in $\set{E}^*_{\max}$. \textbf{(c)} A connected subgraph of $G$ and \textbf{(d)} the corresponding temporal subgraph that is not a valid subgraph because the events of node $c$ are not consecutive: $e_2$ takes place between $e_1$ and $e_3$ and is $\Delta t$-connected to them, so it should be included.}
\end{figure}

Note that the inverse is not true: there are connected subgraphs of $G$ whose vertex sets do not correspond to any valid subgraph; an example is given in Figure \ref{fig:connected_subgraphs}. Therefore to identify all valid subgraphs $\set{E}^* \subseteq \set{E}^*_{\max}$ we first need to find all distinct connected subgraphs of $G$ and then ensure that the corresponding subgraphs are valid by checking that for every node the events (that are in $\set{E}^*_{\max}$ and hence $\Delta t$-connected) are consecutive. This check can be carried out with little extra cost while constructing the colored graph needed to calculate the canonical form, as the construction requires going through all events anyway.

A pseudo-code to identify all connected vertex sets of an arbitrary graph (in this case $G(\set{E}^*_{\max})$) is given in Algorithm \ref{alg:identify_subgraphs}. In function \textsc{FindConnectedSets} we first start $|V|$ search trees so that the tree initialized with node $i$ will find all connected sets where $i$ is the smallest node. The nodes in the set $V_{-}$ are excluded from that search tree; initially this set contains all nodes smaller than $i$. The set $V_{+}$ includes nodes where the search can progress, initially all neighbours larger than $i$. Because each search tree finds only sets where $i$ is the smallest node, the trees are necessarily distinct.

The function \textsc{SubFind} first adds the current set to be returned (line 10) and then grows the sets recursively. For each node $i \in V_{+}$, $V_{-}$ is updated by excluding nodes smaller than $i$. Thus each subtree has a different smallest node from those in $V_{+}$ and the subtrees are again distinct. The set $V_{+}$ is updated to contain nodes where the search can progress: previously allowed nodes larger than $i$ and those new neighbours of $i$ that are not yet excluded.

\begin{algorithm}[tb]
  \caption{\label{alg:identify_subgraphs} Find the vertex sets of all connected subgraphs of a arbitrary graph $G$. The algorithm assumes that nodes are labeled with integers from $1$ to $|V|$. The parameter $n_{\max}$ can be used to limit the size of the vertex sets returned. $N(i)$ denotes the neighbours of node $i$ in graph $G$.}
\begin{algorithmic}[1]
  \Require $G=(V,L)$ is an undirected graph.
\Statex
\Function{FindConnectedSets}{$G$, $n_{\max}$}
  \State $S_{\textrm{all}} \gets \emptyset$
  \For{$i$ in $V$}
    \State $S \gets \{i \}$
    \State $V_{-} \gets \{j \in V\, | \, j \leq i \}$
    \State $V_{+} \gets \{j \in N(i) \, | \, j > i \}$
    \State \Call{SubFind}{$G$, $n_{\max}$, $S_{\textrm{all}}$, $S$, $V_{-}$, $V_{+}$}
  \EndFor
  \State \Return{$S_{\textrm{all}}$}
\EndFunction
\Statex
\Function{SubFind}{$G$, $n_{\max}$, $S_{\textrm{all}}$, $S$, $V_{-}$, $V_{+}$}
  \State $S_{\textrm{all}} \gets S_{\textrm{all}} \cup \{ S \}$
  \If{$|S| = n_{\max}$}
    \Return
  \EndIf
  \For{$i$ in $V_{+}$}
    \State $S^{*} \gets S \cup \{i \}$
    \State $V_{-}^{*} \gets V_{-} \cup \{j \in V_{+}\,|\,j \leq i \}$
    \State $V_{+}^{*} \gets \{j \in V_{+}\,|\, j > i \} \cup \{j \in N(i)\, |\, j \not \in V_{-}^{*} \}$

     \State \Call{SubFind}{$G$, $n_{\max}$, $S_{\textrm{all}}$, $S^{*}$, $V_{-}^{*}$, $V_{+}^{*}$}
  \EndFor
\EndFunction

\end{algorithmic}
\end{algorithm}

Because the subtrees are distinct at each step, the algorithm will return each connected set at most once. To see that it returns all possible connected set, consider how we could arrive at an arbitrary connected set $S$. The search path is rooted at $i_1 = \min S$. Let $S_k$, $k <  |S|$, be the set of elements added at depth $k$. Because $S$ is connected, there is at least one node in $S \backslash S_k$ that is a neighbour of some node in $S_k$. The only way the construction can fail is if for some $k$ there is a node $i^* \in S \backslash S_k$ that has already been excluded, i.e. it is in $V_{-}$. It is not possible that $i^*$ was excluded in the beginning---the tree was rooted at $i_1 < i^*$ and only nodes smaller than $i_1$ were excluded---so it must have happened during the search. But if $i^*$ was added to $V_{-}$ it means that it was in $V_{+}$ but some larger node of $S$ was added instead, which is a contradiction---in the subtree leading to $S$ we would have added $i^*$. Hence no node $i^*$ can exist and the construction can always proceed until $S$ is obtained.

\section{Event labels with partial order}
\label{appendix:labels}

Pseudo-code for finding the labels is presented in Algorithm \ref{alg:label_nodes}. On lines 2--8 we first initialize all labels to $(1,\emptyset)$, unless there are multiple roots which get each initialized with a unique element. The loop on lines 9--19 then iteratively pushes the labels forward along directed paths, adding new elements when needed: first pick any node with zero in-degree (line 10), and find its children who can not be reached through other children ($V_c^-$) (lines 11--13). The labels of these children are then updated by incrementing the value of $r$ and pushing the set $\set{s}_i$ down to the child (line 14--15), adding a new unique element to each set if there are multiple children to update (lines 16--18).

\begin{algorithm}[h!]
  \caption{\label{alg:label_nodes} Find the labels to denote the ordering of events. The vertices in the input graph $G$ correspond to events, and the graph contains at least one directed path from event $e_i$ to $e_j$ if $e_i$ must take place before $e_j$.}
\begin{algorithmic}[1]
  \Require $G=(V,L)$ is a directed acyclic graph.
  \Require $P(i)$ is the set of nodes from which there is a directed path to $v_i$.
  \Statex
\Function{FindEventLabels}{$G$}
\State $s_{\max} \gets 0$
\State $V_0 \gets \{e_i \in V \, | \, k_{\textrm{in}}(e_i) = 0\}$
\For{$e_i$ in $V$}
  \State $r_i \gets 1$, $\set{s}_i \gets \emptyset$
  \If{$|V_0| > 1$ and $e_i \in V_0$}
    \State $\set{s}_i \gets \{ s_{\max} \}$
    \State $s_{\max} \gets s_{\max} + 1$
  \EndIf
\EndFor
\While{$V \neq \emptyset$}
   \State Pick any node $e_i$ with $k_{\textrm{in}}(e_i) = 0$ in $G$
   \State $V_c \gets \{e_j \in V\, | \,(e_i,e_j) \in L \}$
   \State $V_c^- \gets \{e_j \in V_c \, | \, V_c \cap P(j) = \emptyset\}$
   \ForAll{$e_c$ in $V_c^-$}
     \State $r_c \gets \max\{r_c,r_i+1\}$
     \State $\set{s}_c \gets \set{s}_c \cup \set{s}_i$
     \If{$|V_c^+| > 1$}
        \State $\set{s}_c \gets \set{s}_c \cup \{ s_{\max} \}$
        \State $s_{\max} \gets s_{\max} + 1$
     \EndIf
   \EndFor
   \State Remove $e_i$ from $G=(V,L)$
\EndWhile
\EndFunction
\end{algorithmic}
\end{algorithm}

It is easy to see that this algorithm results in valid labels:
\begin{itemize}
\item If event $e_i$ must come before $e_j$ then there is (at least one) directed path from $e_i$ to $e_j$. The set $\set{s}_i$ will be   pushed along this path and therefore $\set{s}_i \subseteq   \set{s}_j$, and the value $r$ will increase along this path and   therefore $r_i < r_j$.
\item If there is no restriction for the mutual order of $e_i$ and   $e_j$, there is no directed path between these nodes. We can trace   backwards all paths from these nodes to find that either
\begin{enumerate}
\item the paths traced from $e_i$ and from $e_j$ end up in distinct   root nodes and hence $s_i \cap s_j = \emptyset$ because each root   was initialized with a different element,
\item the paths converge to a common node that hence has multiple   children and the two branches were assigned unique elements of   $\set{S}$, and therefore $s_i \not \subseteq s_j$ and $s_j \not   \subseteq s_i$, or
\item both of the above (there can be multiple directed paths leading to these nodes), which also means that $s_i \not \subseteq s_j$ and $s_j \not \subseteq s_i$.
\end{enumerate}
\end{itemize}
The labels this algorithm provides are plain integers when the input graph $G$ contains a total order of events.

\section*{Bibliography}

\end{document}